\documentclass[runningheads]{llncs}
\usepackage{amssymb}
\usepackage{amsmath}
\usepackage{comment}
\usepackage[disable]{todonotes}
\usepackage[utf8]{inputenc}
\usepackage{shuffle}
\usepackage{multirow}
\usepackage{listings}
\usepackage{mathabx}
\usepackage{ mathrsfs }

\usepackage{apxproof}
\theoremstyle{plain}
\newtheoremrep{theorem}{Theorem}[section]
 \newtheoremrep{proposition}[theorem]{Proposition}
\newtheoremrep{lemma}[theorem]{Lemma}
 \newtheoremrep{claim}[theorem]{Claim}
 \newtheoremrep{conjecture}[theorem]{Conjecture}
 \newtheoremrep{corollary}[theorem]{Corollary}
 
\theoremstyle{definition}
\newtheoremrep{definition}[theorem]{Definition} 
 
\usepackage{hyperref}

\usepackage{diagbox}
\usepackage{slashbox}
\usepackage{tikz}
\usetikzlibrary{arrows,shapes,automata,positioning}

\newcommand{\stc}{\operatorname{sc}}

\newcommand{\Syn}{\operatorname{Syn}}

\def\dotminus{\mathbin{\ooalign{\hss\raise1ex\hbox{.}\hss\cr
  \mathsurround=0pt$-$}}} 
 
\begin{document}
\title{Completely Reachable Automata, Primitive Groups and the State Complexity of the Set of Synchronizing Words}
\titlerunning{Completely Reachable Automata, Primitive Groups and State Complexity}

\author{Stefan Hoffmann\orcidID{0000-0002-7866-075X}}
\authorrunning{S. Hoffmann}
\institute{Informatikwissenschaften, FB IV, 
  Universit\"at Trier,  Universitätsring 15, 54296~Trier, Germany, 
  \email{hoffmanns@informatik.uni-trier.de}}
\maketitle              % typeset the header of the contribution
\begin{abstract}

  We give a new characterization of primitive permutation groups
  tied to the notion of completely reachable automata.
  Also, we introduce sync-maximal permutation groups tied
  to the state complexity of the set of synchronizing words of certain associated automata
  and show that they are contained between the $2$-homogeneous
  and the primitive groups.
  %Furthermore,  
  %for automata
  %containing groups and for automata over a binary alphabet,
  %we state several sufficient criteria which
  %give %synchronizing 
  %automata for which the state complexity
  %of the set of synchronizing words is maximal.
  %One such criterion applies 
  %These apply 
  %to a family of automata
  %from the literature for which this was only conjectured.
  %Hence, we solve an open open and give a wealth of additional 
  %automata with this property. 
  Lastly, we define $k$-reachable %completely reachable
  groups in analogy with synchronizing groups and motivated by our characterization
  of primitive permutation groups. 
  %But the results show
  %that only the alternating and the symmetric groups
  %arise. %are $k$-reachable.% in this sense.
  But the results show that a $k$-reachable permutation group of degree $n$
  with $6 \le k \le n - 6$ is either the alternating or the symmetric group.
  
   %
  % aperidoc 1-contracting sind doch die binären mit zyklus und rk n -1
  % siehe don paper, also gilt doch äquvalenz, und damit compl reachable in P entscheidbar
  % zu beiden mit look-up table image untersuchen, bzw. bahn berechnen.
 
\end{abstract}

\section{Introduction}

\todo{In Appendix, den Beweis von rank $n-1$, dann permutation raus. Ist für anderes paper.}

A deterministic semi-automaton is synchronizing if it admits a reset word, i.e., a word which leads to some definite
state, regardless of the starting state. This notion has a wide range of applications, from software testing, circuit synthesis, communication engineering and the like, see~\cite{San2005,Vol2008}.  The famous \v{C}ern\'y conjecture~\cite{Cer64}
states that a minimal length synchronizing word has length at most $(n-1)^2$
for an $n$ state automaton.
We refer to the mentioned survey articles~\cite{San2005,Vol2008} for details. 
%Due to its importance, the notion of synchronization has undergone a range of generalizations and variations
%for other automata models.
An automaton is completely reachable, if for each subset of states 
we can find a word which maps the whole state set onto this subset.
%This is a stronger requirement than synchronizability,
This is a generalization of synchronizability,
as a synchronizing word maps the whole state set to a singleton set.
%and surely we can find synchronizing automata which are not completely reachable.
The class of completely reachable automata was formally introduced in~\cite{DBLP:conf/dcfs/BondarV16},
but already in~\cite{DBLP:journals/combinatorics/Don16,DBLP:journals/corr/Maslennikova14}
such automata appear in the results.
The time complexity of deciding if a given automaton is completely reachable
is unknown. A sufficient and necessary criterion for complete reachability of a given automaton
in terms of graphs and their connectivity
is known~\cite{DBLP:conf/dlt/BondarV18}, but it is not known if these graphs
could be constructed in polynomial time.
%A special case, the graph constructed by considering words that collaps precisely two states,
%which gives a sufficient criterion for complete reachability~\cite{DBLP:conf/dcfs/BondarV16},
%is shown to be constructible in polynomial time~\cite{DBLP:journals/jalc/GonzeJ19}.
A special case of the general graph construction, which gives a sufficient
criterion for complete reachability~\cite{DBLP:conf/dcfs/BondarV16},
is known to be constructible in polynomial time~\cite{DBLP:journals/jalc/GonzeJ19}.
The size of a minimal automaton accepting a given regular
language is called the state complexity of that language.
The set of synchronizing words of a given automaton is a regular ideal language
whose state complexity is at most exponential in the size of the original 
automaton~\cite{DBLP:journals/corr/Maslennikova14,DBLP:journals/ijfcs/Maslennikova19}. 
The \v{C}ern\'y family
of automata~\cite{Cer64,Vol2008}, the first given family yielding the lower bound $(n-1)^2$ for the length of synchronizing words,
is completely reachable and the corresponding sets of synchronizing words have maximal state 
complexity~\cite{DBLP:journals/corr/Maslennikova14,DBLP:journals/ijfcs/Maslennikova19}. 
%These properties are shared by many families
%of automata that are also slowly synchronizing~\cite{DBLP:conf/mfcs/AnanichevGV10,AnaVolGus2013,DBLP:journals/corr/Maslennikova14,DBLP:journals/ijfcs/Maslennikova19}.
The notion of primitive permutation groups could be traced back to work by Galois~\cite{Neumann2011} on the solubility of 
equations by radicals. Nowadays, it is a core notion of the theory
of permutation groups~\cite{cameron_1999}.
% We give a new characterization of primitive groups which relates these groups
% to completely reachable automata. 
% Furthermore, we investigate the relation between the state complexity of
% the set of synchronizing words and complete reachability. 
% We state a general criterion and a polynomial time algorithm 
% to decide for a given completely reachable automaton if the set of synchronizing
% words has maximal state complexity.
% We then take a closer look at automata over a binary automata and
% give results, that together with a result by~\cite{DBLP:journals/combinatorics/Don16},
% allows the construction of many new classes of automata
% which are completely reachable and such that the state
% complexity of the set of synchronizing words is maximal.
% Our results also encompass a class from~\cite{DBLP:journals/ijfcs/Maslennikova19} for which this was only conjectured,
% hence settling this problem.

% Motivation noch dazu?
\subsubsection{Outline and Contribution:} In Section~\ref{sec::preliminaries}, we give definitions
and state known results. Then, in Section~\ref{sec:gen_results},  for completely
reachable automata, we state a sufficient and necessary condition for the set of synchronizing words to have maximal state complexity, which yields a polynomial time decision procedure.
In Section~\ref{sec:grp_case}, 
we introduce a new characterization
of primitive permutation groups, motivated by work on synchronizing and
completely reachable automata and on detecting properties of permutation groups
by functions~\cite{ArCamSt2017,DBLP:journals/tcs/ArnoldS06,DBLP:conf/dcfs/BondarV16,DBLP:journals/combinatorics/Don16}. We relate this to the notion of the state complexity
of the set of synchronizing words. 
Beside the \v{C}ern\'y family~\cite{Cer64,Vol2008},
the properties that a minimal length synchronizing word has quadratic length
and that the set of synchronizing words has maximal state complexity are shared by a wealth of different slowly synchronizing
automata~\cite{DBLP:conf/mfcs/AnanichevGV10,AnaVolGus2013,DBLP:journals/corr/Maslennikova14,DBLP:journals/ijfcs/Maslennikova19}. %In connection with this, we introduce the class
Motivated by this,  we introduce the class
of sync-maximal permutation groups and show that they fit properly between the $2$-homogeneous
and the primitive groups.
%This criterion yields a polynomial time algorithm to decide this property for completely 
%reachable automata. 
%In Section~\ref{sec:binary_case}, we apply the results from Section~\ref{sec:gen_results}
%to give more specific sufficient criteria. We apply these to solve an open problem;
%we confirm a conjecture concerning the maximal state complexity of the set of synchronizing
%words for a family of automata from~\cite{DBLP:journals/ijfcs/Maslennikova19}.
Lastly, in Section~\ref{sec:k_reachable_groups} we introduce $k$-reachable groups motivated
by our investigations and the definition of synchronizing groups~\cite{DBLP:journals/tcs/ArnoldS06}.
We show that for almost all $k$, only the symmetric and alternating groups are $k$-reachable.

\section{Preliminaries and Definitions}
\label{sec::preliminaries}

% todo noch überall \emph{} einfügen

\subsubsection{General Notions:} Let $\Sigma = \{a_1,\ldots, a_k\}$ be a finite set of symbols,
%\footnote{If not otherwise stated we assume that our alphabet has
%the form 
% $\Sigma = \{a_1, \ldots, a_k\}$ and $k$ denotes the number
% of symbols.},
 called an \emph{alphabet}. 
 By $\Sigma^{\ast}$, we denote
the \emph{set of all finite sequences}, i.e., of all words or strings. The \emph{empty word}, i.e., the 
finite sequence of length zero, is denoted by $\varepsilon$. We set $\Sigma^+ = \Sigma^* \setminus \{\varepsilon\}$.
%We call a word $u \in \Sigma^*$ a \emph{prefix} of another word $v \in \Sigma^*$,
%if there exists $x \in \Sigma^*$ such that $v = ux$.
For a given word $w \in \Sigma^*$, we denote by $|w|$
its \emph{length}.
%, and for $a \in \Sigma$ by $|w|_a$ the number of occurrences of the symbol $a$
%in $w$. 
The subsets of $\Sigma^{\ast}$
are called \emph{languages}. 
%With $\mathbb N_0 = \{ 0,1,2,\ldots \}$
%we denote the set of natural numbers, including zero, and $\mathbb N = \mathbb N_0 \setminus \{0\}$.
%By $\mathbb Z = \{\ldots, -3,-2,-1,0,1,2,3,\ldots \}$ we denote the integers.
For $n > 0$, we set $[n] = \{0,\ldots, n-1\}$ and $[0] = \emptyset$.
%If $a,b \in \mathbb Z$ and $b \ne 0$, by $a\bmod b$ we denote the unique number $0 \le r < b$
%with $a = qb + r$ for some $q \in \mathbb Z$.
For a set $X$, we denote the \emph{power set} of $X$ by $\mathcal P(X)$, i.e,
the set of all subsets of $X$. Every function $f : X \to Y$
induces a function $\hat f : \mathcal P(X) \to \mathcal P(Y)$
by setting $\hat f(Z) := \{ f(z) \mid z \in Z \}$.
Here, we will denote this extension also by~$f$.
Let $k \ge 1$. A \emph{$k$-subset} $Y \subseteq X$ is a finite set of cardinality~$k$.
A $1$-set is also called a \emph{singleton set}.
For functions $f : A \to B$ and $g : B \to C$, the \emph{functional composition}
$gf : A \to C$ is the function $(gf)(x) = g(f(x))$, i.e., the function on the right is applied first\footnote{In group theory, usually the other convention is adopted, but we stick to the convention most often seen in formal
language theory.}.

% A bijective function $f : X \to X$ is called a \emph{permutation (of $X$)}.
% If $X$ is finite, the smallest number $n \ge 0$ such that
% $f^n(x) = x$ for each $x \in X$ is called the \emph{order of the permutation}.
% Every permutation could be written uniquely as a composition of disjoint 
% cycles~\cite{cameron_1999}, then the order equals the least common multiple of all the resulting cycle
% lengths.
% By $\pi_1 : X\times Y \to X$ and $\pi_2 : X \times Y \to Y$ we denote
% the projection maps onto the first and second component, $\pi_1(x,y) = x$ and $\pi_2(x,y) = y$.
% If $a,b\in \mathbb N_0$ with $b > 0$, we denote by $a \bmod b$ the unique number $0 \le r < b$
% such that $a = bn + r$ for some $n \ge 0$.
%For $n \in \mathbb N_0$ we set $[n] := \{ k \in \mathbb N_0 : 0 \le k < n \}$.
%Let $M \subseteq \mathbb N_0$ be some \emph{finite} set. 
%By $\max M$ we denote the maximal element in $M$
%with respect to the usual order, and we set $\max \emptyset = 0$.
% Also for finite $M \subseteq \mathbb N_0 \setminus\{0\}$, i.e., $M$ is finite without zero in it,
% by $\lcm M$ we denote the
% least common multiple of the numbers in $M$, and set $\lcm \emptyset = 0$.
\subsubsection{Automata-Theoretic Notions:} A \emph{finite, deterministic} and \emph{complete} automaton will
be denoted by $\mathscr A = (\Sigma, Q, \delta, s_0, F)$
with $\delta : Q \times \Sigma \to Q$ the state transition function, $Q$ a finite set of states, $s_0 \in Q$
the start state and $F \subseteq Q$ the set of final states. 
The properties of being deterministic and complete are implied by the definition of $\delta$
as a total function.
The transition function $\delta : Q \times \Sigma \to Q$
could be extended to a transition function on words $\delta^{\ast} : Q \times \Sigma^{\ast} \to Q$
by setting $\delta^{\ast}(s, \varepsilon) := s$ and $\delta^{\ast}(s, wa) := \delta(\delta^{\ast}(s, w), a)$
for $s \in Q$, $a \in \Sigma$ and $w \in \Sigma^{\ast}$. In the remainder we drop
the distinction between both functions and will also denote this extension by $\delta$.
For $S \subseteq Q$ and $w \in \Sigma^*$, we write $\delta(S, w) = \{ \delta(s, w) \mid s \in S \}$
and $\delta^{-1}(S, w) = \{ q \in Q \mid \delta(q, w) \in S \}$.
%Here we do not use
%other automata models.
% hier smei-auto einführen! todo
% A \emph{semi-automaton} $\mathcal A = (\Sigma, Q, \delta)$ 
% is like an automaton, but without a designated start state and without a set of final states.
% All (semi-)automata considered in this paper will be complete and deterministic. 
%and initially connected, the last notion meaning
%for every $s \in S$ there exists some $w \in \Sigma^{\ast}$ such that $\delta(s_0, w) = s$.
%Let $\mathcal A = (\Sigma, Q, \delta, s_0, F)$ 
% be an automaton, then a trap state $t \in Q$ 
% is a state such that $\delta(t,x) = t$ for each $x \in \Sigma$.
The \emph{language accepted}, or \emph{recognized}, by $\mathscr A = (\Sigma, S, \delta, s_0, F)$ is
$
 L(\mathscr A) = \{ w \in \Sigma^{\ast} \mid \delta(s_0, w) \in F \}.
$
A language $L \subseteq \Sigma^{\ast}$ is called \emph{regular} if $L = L(\mathscr A)$
for some finite automaton $\mathscr A$. 
For a language $L \subseteq \Sigma^{\ast}$ and $u, v \in \Sigma^*$
we define the \emph{Nerode right-congruence} 
with respect to $L$ by $u \equiv_L v$ if and only if
$
 \forall x \in \Sigma : ux \in L \leftrightarrow vx \in L.
$
The equivalence class for some $w \in \Sigma^{\ast}$
is denoted by $[w]_{\equiv L} := \{ x \in \Sigma^{\ast} \mid x \equiv_L w \}$.
A language is regular if and only if the above right-congruence has finite index, and it could
be used to define the minimal deterministic automaton
$
 \mathscr A_L = (\Sigma, Q, \delta, [\varepsilon]_{\equiv_L}, F)
$
with $Q := \{ [w]_{\equiv_L} \mid w \in \Sigma^{\ast} \}$, $\delta([w]_{\equiv_L}, a) := [wa]_{\equiv_L}$
for $a \in \Sigma$, $w \in \Sigma^{\ast}$ and $F := \{ [w]_{\equiv_L} \mid w \in L \}$.
It is indeed the smallest automaton accepting $L$ in terms of states, and we 
will refer to this construction as the minimal automaton~\cite{HopUll79} of $L$.
The \emph{state complexity} of a regular language
is defined as the number of Nerode right-congruence classes.
We will denote this number by $\stc(L)$.
%For $u \in \Sigma^*$ and $L \subseteq \Sigma^*$, the \emph{quotient}
%of $L$ by $u$ is the set $u^{-1}L = \{ v \in \Sigma^* \mid uv \in L \}$.
%If $u,v \in \Sigma^*$, we have $u \equiv_L v$ if and only if $u^{-1}L = v^{-1}L$.
%Hence, the quotients could be used as representatives of the Nerode equivalence classes.
%Indeed, the automaton $\mathcal A = (\Sigma, Q, \delta, L, F)$
%with $Q = \{ u^{-1}L \mid u \in \Sigma^* \}$, $F = \{ u^{-1}L \mid u \in \Sigma^*, \varepsilon \in u^{-1}L \}$
%and $\delta(u^{-1}L, a) = (ua)^{-1}L$
%is isomorphic to the minimal automaton of $L$.
%Hence, $\stc(L) = |\{ u^{-1}L \mid u \in \Sigma^* \}|$ for regular $L \subseteq \Sigma^*$.
Let $\mathscr A = (\Sigma, Q, \delta, s_0, F)$ be an automaton.
A state $q \in Q$ is \emph{reachable}, if $q = \delta(s_0, u)$
for some $u \in \Sigma^*$. We also say that a state $q$ is reachable from a state $q'$
if $q = \delta(q', u)$ for some $u \in \Sigma^*$.
Two states $q, q'$ are \emph{distinguishable},
if there exists $u \in \Sigma^*$ such that either $\delta(q, u) \in F$ and $\delta(q', u) \notin F$
or $\delta(q, u) \notin F$ and $\delta(q', u) \in F$.
An automaton for a regular language is isomorphic to the minimal automaton
if and only if all states are reachable and distinguishable~\cite{HopUll79}.
A \emph{semi-automaton} $\mathscr A = (\Sigma, Q, \delta)$
is like an ordinary automaton, but without a designated start state and without a set of final 
states. Sometimes, we will also call a semi-automaton simply an automaton if the context
makes it clear what is meant. Also, definitions without explicit reference
to a start state and a set of final states are also valid for semi-automata.
Let $\mathscr A = (\Sigma, Q, \delta)$ be a finite semi-automaton.
A word $w \in \Sigma^*$ is called \emph{synchronizing}, if $\delta(q, w) = \delta(q', w)$
for all $q, q' \in Q$, or equivalently $|\delta(Q, w)| = 1$.
Set $\Syn(\mathscr A) = \{ w \in \Sigma^* \mid |\delta(Q, w)| = 1 \}$.
The \emph{power automaton (for synchronizing words)}
associated to $\mathscr A$ is $\mathcal P_{\mathscr A} = (\Sigma, \mathcal P(Q), \delta, Q, F)$
with start state $Q$, final states $F = \{ \{q \} \mid q\in Q \}$ and the transition function of $\mathcal P_{\mathscr A}$
is the transition function of $\mathscr A$, but applied to subsets of states.
Then, as observed in~\cite{Starke66a},
the automaton $\mathcal P_{\mathscr A}$ accepts the set of synchronizing words,
i.e., $L(\mathcal P_{\mathscr A}) = \Syn(\mathscr A)$.
As for $\{q\} \in F$, we also have $\delta(\{q\}, x) \in F$ for each $x \in \Sigma^*$,
the states in $F$ could all be merged to a single state to get an accepting automaton
for $\Syn(\mathscr A)$.
Also, the empty set is not reachable from $Q$. Hence $\stc(\Syn(\mathscr A)) \le 2^{|Q|} - |Q|$ 
and this bound is sharp~\cite{DBLP:journals/corr/Maslennikova14,DBLP:journals/ijfcs/Maslennikova19}.
We call $\mathscr A$ \emph{completely reachable}, if for any non-empty
$S \subseteq Q$, there exists a word $w \in \Sigma^*$ with $\delta(Q, w) = S$, i.e.,
in the power automaton, every state is reachable from the start state. When
we say a \emph{subset of states} in $\mathscr A$ is \emph{reachable}, we mean reachability in $\mathcal P_{\mathscr A}$.
The state complexity of $\Syn(\mathscr A)$ is maximal, i.e., $\stc(\Syn(\mathscr A)) = 2^{|Q|} - |Q|$,
if and only if all subsets $S \subseteq Q$ with $|S| \ge 2$ are reachable
and at least one singleton subset of $Q$, and all these states are distinguishable in $\mathcal P_{\mathscr A}$.
For \emph{strongly connected automata}, i.e., those for which all states are reachable
from each other, the state complexity of $\Syn(\mathscr A)$
is maximal iff $\mathscr A$ is completely reachable and all $S \subseteq Q$ with $|S| \ge 2$
are distinguishable %\footnote{The singleton are distinguishable
%from all subsets with at least two distinct elements.} 
in $\mathcal P_{\mathscr A}$.% (recall that $F = \{ \{q\} \mid q \in Q \}$, hence $|\{\{q\},S\}\cap F| = 1$). %, as the singleton set are precisely the final states these are always distinguishable
%from every set with at least two states.

\subsubsection{Transformations and Permutation Groups:} Let $n \ge 0$. %Set $[n] = \{ 0,1, \ldots, n-1 \}$.
Denote by $\mathcal S_n$ the \emph{symmetric group on $[n]$}, i.e., the group of all permutations of $[n]$.
A \emph{permutation group (of degree $n$)}
is a subgroup of $\mathcal S_n$.
For $n > 1$, the \emph{alternating group} is the unique subgroup of size $n! / 2$
in $\mathcal S_n$, see~\cite{cameron_1999}.
The \emph{orbit} of an element $i \in [n]$ for a permutation group $G$ is the set $\{ g(i) \mid g \in G\}$.
A permutation group $G$ over $[n]$ is \emph{primitive}, if it preserves no non-trivial equivalence
relation\footnote{The trivial equivalence relations on $[n]$ are $[n] \times [n]$
and $\{ (x,x) \mid x \in [n] \}$.}
on $[n]$, i.e., for no non-trivial equivalence relation $\sim \subseteq [n]\times [n]$
we have $p \sim q$ if and only if $g(p) \sim g(q)$ for all $g \in G$ and $p, q \in [n]$.
A permutation group $G$ over $[n]$ is called \emph{k}-homogeneous for some $k \ge 1$,
if for any two $k$-subsets $S,T$ of $Q$, there exists $g \in G$
such that $g(S) = T$. A \emph{transitive} permutation group is the same
as a $1$-homogeneous permutation group.
Note that here, all permutation groups with $n \le 2$ are primitive, and for $n > 2$
every primitive group is transitive. Because of this, some authors exclude
the trivial group for $n = 2$ from being primitive.
A permutation group $G$ over $[n]$ is called \emph{$k$-transitive} for some $k \ge 1$,
if for any two tuples $(p_1, \ldots, p_k), (q_1, \ldots, q_k) \in [n]^k$,
there exists $g \in G$
such that $(g(p_1),\ldots, g(p_k)) = (q_1, \ldots, g_k)$.
By $\mathcal T_n$, we denote the set of all maps on $[n]$.
A submonoid of $\mathcal T_n$ for some $n$ is called a \emph{transformation monoid}.
If the set $U$ is a submonoid (or a subgroup) of $\mathcal T_n$ (or $\mathcal S_n$)
we denote this by $U \le \mathcal T_n$ (or $U \le \mathcal S_n$).
For a set $A \subseteq \mathcal T_n$ (or $A \subseteq \mathcal S_n$),
we denote by $\langle A \rangle$ the submonoid (or the subgroup) generated
by $A$. Let $\mathscr A = (\Sigma, Q, \delta)$ be an semi-automaton and
for $w \in \Sigma^*$ define $\delta_w : Q \to Q$
by $\delta_w(q) = \delta(q,w)$ for all $q \in Q$.
Then, we can associate with $\mathscr A$ the transformation monoid
of the automaton $\mathcal T_{\mathscr A} = \{ \delta_w \mid w \in \Sigma^* \}$, where we can identify $Q$
with $[n]$ for $n = |Q|$.
%Seeing the letters as transformations of the state set,
%$\mathcal T_{\mathscr A} = \langle \Sigma \rangle$.
We have $\mathcal T_{\mathscr A} = \langle \{ \delta_x \mid x \in \Sigma \} \rangle$.
The \emph{rank} of a map $f : [n] \to [n]$ is the cardinality of its image.
For a given semi-automaton $\mathscr A = (\Sigma, Q, \delta)$, the \emph{rank of a word} $w \in \Sigma^*$
is the rank of $\delta_w$.

\subsubsection{Known Results:} 
% We need the following result from~\cite{DBLP:journals/combinatorics/Don16}.

% % notationen festlegen und sprechweisen, todo
% \begin{proposition}\cite{DBLP:journals/combinatorics/Don16}
% \label{prop:don_compl_reach}
%  Let $\mathscr A = (\Sigma, Q, \delta)$
%  be a finite automaton with $n$ states. Suppose we have two letters $a,b \in \Sigma$
%  such that $a$ has rank $n-1$
%  and $b$ is a permutation with a single orbit.
%  Choose $s,t \in Q$ and $0 < d < |Q|$ such that $\delta(Q, a) = Q \setminus \{s\}$, 
%  $|\delta^{-1}(t, a)| = 2$ and $\delta(s, b^d) = t$.
%  If $d$ and $n$ are coprime, then
% %  \begin{enumerate}
% %  %\item $\mathscr A$ is synchronizing. % Naja, eigentlich impliziert 3 die oberen beiden, also weglassen?
% % 
% %  %\item The shortest synchronizing word of $\mathscr A$ has length at most $(n-1)^2$.
% % 
% %  \item For every non-empty set $S \subseteq Q$ of size $k$, there exists
% %  a word $w_S$ of length at most $n(n-k)$ such that $\delta(Q, w_S) = S$.
% %  \end{enumerate}
%  for every non-empty set $S \subseteq Q$ of size $k$, there exists
%  a word $w_S$ of length at most $n(n-k)$ such that $\delta(Q, w_S) = S$.
% \end{proposition}
%
The next result appears in~\cite{ArCamSt2017}
and despite it was never clearly spelled out by Rystsov himself, it is implicitly present in arguments 
used in~\cite{Rystsov2000}.

\begin{theorem}[Rystsov~\cite{ArCamSt2017,Rystsov2000}]
\label{thm:rystsov} 
 A permutation group $G$ on $[n]$ 
 is primitive if and only if, for any map $f : [n] \to [n]$
 of rank $n - 1$, the transformation monoid
 $\langle G \cup \{f\} \rangle$ contains a constant map.
\end{theorem}

In~\cite{DBLP:conf/dcfs/BondarV16} a sufficient criterion for complete
reachability was given. It is based on the following graph construction associated to
a semi-automaton.

\begin{definition}[Bondar \& Volkov \cite{DBLP:conf/dcfs/BondarV16}]
\label{def:Gamma_1}
%  Let $M \subseteq \mathcal T_n$ be a transformation monoid.
%  Then, we define the $\Gamma_1(\mathscr A) = ([n], E)$
%  with vertex set $[n]$ and edge set
%  $$
%   E = \{ (p, q) \mid \exists f \in M : p \notin f([n]), |f^{-1}(q)| = 2, \mbox{ $f$ has rank $n - 1$} \}.
%  $$
 Let $\mathscr A = (\Sigma, Q, \delta)$ be a semi-automaton.
 Then, we define the graph $\Gamma_1(\mathscr A) = (Q, E)$
 with vertex set $Q$ and edge set
 $
  E = \{ (p, q) \mid \exists w \in \Sigma^* : p \notin \delta(Q, w), |\delta^{-1}(q, w)| = 2, \mbox{ $w$ has rank $|Q| - 1$} \}.
 $
 For transformation monoids $M \le \mathcal T_n$, a similar definition $\Gamma_1(M)$ applies.
%  A similar definition could be given for any transformation monoid $M \subseteq \mathcal T_n$,
%  denotes by $\Gamma_1(M)$.
\end{definition}

The construction was extended in~\cite{DBLP:conf/dlt/BondarV18}
to give a sufficient and necessary criterion. The graph $\Gamma_1(\mathscr A)$ could be computed in polynomial time~\cite{DBLP:journals/jalc/GonzeJ19}.

\begin{theorem}[Bondar \& Volkov \cite{DBLP:conf/dcfs/BondarV16}]
\label{thm:Gamma_1_scc}
 Let $\mathscr A = (\Sigma, Q, \delta)$.  If $\Gamma_1(\mathscr A) = (Q, E)$
 is strongly connected, then $\mathscr A$
 is completely reachable.
\end{theorem}

% singleton set defineiren

%\section{Characterizing Primitive and $2$-Homogenous Permutation Groups}

\section{General Results on the State Complexity of $\Syn(\mathscr A)$}
\label{sec:gen_results}

The first result of this section will be needed later when we investige the properties of the sync-maximal
groups introduced in Section~\ref{sec:sync_max}.
But it could also be used for a polynomial time decision
procedure to decide if, for a given completely reachable automaton,
the set of synchronizing words has maximal state complexity.

\begin{lemmarep} 
\label{lem:compl_reach_implies_max_Syn_2sets}
 %Let $\mathscr A = (\Sigma, Q, \delta)$ be completely reachable
 Let $\mathscr A$ be completely reachable
 with $n$ states. Then, $\stc(Syn(\mathscr A)) = 2^n - n$
 if and only if all $2$-sets of states are pairwise distinguishable in $\mathcal P_{\mathscr A}$.
\end{lemmarep} 
\begin{proof}
 In $\mathcal P_{\mathscr A}$, for the singleton sets 
 we have, for each $w \in \Sigma^*$, that $|\delta(\{q\}, w)| = 1$.
 Hence, no two singleton sets could be distinguished in $\mathcal P_{\mathscr A}$
 as the final state set are precisely the singleton subsets of the states.
 Two states that could not be distinguished are also called equivalent, and
 could be merged to give an automaton accepting the same language~\cite{HopUll79}.
 Looking at the definition of $\mathcal P_{\mathscr A}$, we find, for completely reachable $\mathscr A$,
 that $\stc(\Syn(\mathscr A)) = 2^n - n$ if and only if all non-empty non-singleton sets, i.e.,
 those with at least two distinct elements, are distinguishable in $\mathcal P_{\mathscr A}$.
 We will show that the latter condition is equivalent to the fact that we
 could distinguish all $2$-sets of states.
%  The singleton sets $\{q\}$, $q \in Q$, are all equivalent for the language $\Syn(\mathscr A)$,
%  and could be merged to a single final trap state. Any state $S \subseteq Q$ in the power
%  automaton $\mathcal P_{\mathscr A}$ for $\Syn(\mathscr A)$
%  which is not a singleton set is non-final, hence distinguishable from the singleton sets.
%  Hence, it is enough to establish that all non-singleton sets are distinguishable.
%  %We will show by induction on $k$ that the non-singleton sets $A,B$
%  %are distinguishable if $\max\{|A|, |B|\} \le k$.

 If $\stc(\Syn(\mathscr A)) = 2^n - n$, then in particular all $2$-sets
 are distinguishable. Conversely, assume all $2$-subsets of $Q$
 are distinguishable. We prove inductively
 that for $2 \le k \le n$ all distinct $k$-subsets of states are distinguishable,
 where $k$ is our induction parameter. Our assumption gives the base case $k = 2$.
 Suppose the statement is true for $k \ge 2$.
 Let $A, B$ be two non-empty non-singleton sets with $\max\{|A|, |B|\} \le k + 1$.
 We can assume $\max\{|A|, |B|\} = k + 1$, for otherwise the induction hypothesis
 gives the claim.
 We distinguish two cases for the cardinalities of $A$ and $B$.
 
 \begin{enumerate}
 \item[(i)] $|B| < |A| = k + 1$.
  
   \medskip 
   
   Then $|A| \ge 3$. Choose two distinct $2$-sets
   $\{q_1, q_2\}, \{q_3, q_4\} \subseteq A$.
   By hypothesis, the $2$-sets are distinguishable, so 
   %(more specifically, the case of $2$-sets
   %was shown in the base case), 
   we find a word $w \in \Sigma$
   such that one is mapped to a singleton, but not the other.
   Then $|\delta(A, w)| < |A|$. And as $|\delta(\{q_1, q_2, q_3, q_4\}, w)| \ge 2$,
   the set $\delta(A, w)$ is not a singleton set.
   If $\delta(B, w)$ is a singleton, then $w$ distinguishes $A$ and $B$.
   Otherwise $2 \le |\delta(B, w)| \le |B| \le k$
   and we can apply the induction hypothesis to $\delta(A, w)$ and $\delta(B,w)$.
   So, if $u$ distinguishes these two sets, then $wu$
   distinguishes $A$ and $B$.
   
   \medskip 
   
 \item[(ii)] $|B| = |A| = k + 1$.
 
  \medskip 
   
   By the same argument as in (i), we find $w \in \Sigma^*$
   such that $2 \le |\delta(A, w)| \le k$.
   If $\delta(B,w)$ is a singleton, then both sets are distinguished.
   Otherwise, if $|\delta(B,w)| = k + 1$, we continue with
   $\delta(A,w)$ and $\delta(B,w)$ as in case (i).
   If $2 \le |\delta(B,w)| \le k$, 
   then we apply the induction hypothesis.
 \end{enumerate}
 Working our way up to $k = |Q|$, we see that we can distinguish all non-empty non-singleton
 subsets of $Q$ in $\mathcal P_{\mathscr A}$.
 As $\mathscr A$ is assumed to be completely reachable, the minimal automaton
 has $2^n - n$ states, i.e., $\stc(\Syn(\mathscr A)) = 2^n - n$.~$\qed$
\end{proof}

Hence, we only need to check for all pair states $\{ p, q \}$
with $p \ne q$ in the power automaton if they are all distinguishable to each other.
This could be done in polynomial time.

\begin{corollaryrep}
\label{cor:poly_alg_max_sync}
 Let $\mathscr A = (\Sigma, Q, \delta)$ be a completely reachable semi-automaton
 with $n$ states. Then, we can decide in polynomial time
 if $\stc(\Syn(\mathscr A)) = 2^n - n$.
\end{corollaryrep}
\begin{proof}
%  damit poly alg fürs entscheiden, da man nur paare (in tabelle, wie normaler minimierungsalgo)
% unterschieden muss, skizzieren.
 
%  This is essentially the minimization algorithm from~\cite{Hopcroft+Ullman/79/Introduction},
%  adapted to our setting. 
%  In general, if $\mathscr A = (\Sigma, Q, \delta, s_0, F)$
%  is a finite automaton, by the pigenhole principle, if $s, t \in Q$
%  are distinguishable, we have a word of length smaller than $|Q|$
%  that distinguishes them, which guarantees termination
%  of the usual minimization algorithm. %todo prüfen, ob das so stimmt. bzw wenn sich in tablle nichts ändert.
%  %fussnote actually mit tabelle bwz eine äquivalenzrelation stabil wird.
%  % wirklich? muss ja nicht an gleicher stelle "rauschneidbar" sein.
 
%  %
%  % allgemein formulieren, und dann auf P_2 anwenden todo.
%  %
%  % todo sprechweise reank of word
%  Now suppose $\mathscr A = (\Sigma, Q, \delta)$ is completely reachable.
%  Construct the semi-automaton $\mathcal P_2 = (\Sigma, \{ \{q, q'\} \mid q,q'\in Q\}, \delta)$,
%  i.e. we only take the $2$-subsets and the singleton sets of $Q$ as states.
%  Note that, as $|\delta(S, w)| \le |S|$ for any $S \subseteq Q$,
%  the semi-automaton is complete, i.e., the definition makes sense.
%  Set $F_1 = \{ \{ \{q\}, \{q', q''\} \} \mid q, q', q'' \in Q \}$. %First, we mark the sets in $F_1$.
%  Then, we set
%  $$
%   F_{i+1} = \{ \{ \{q_1, q_2 \}, \{ q_3, q_4\} \} \mid 
%                 \{ \{ q_1, q_2 \}^x, \{q_3, q_4\}^x \} \in F_i \mbox{ for some $x\in\Sigma$} \} \cup F_i.
%  $$

 Suppose $\mathscr A = (\Sigma, Q, \delta)$ is completely reachable.
 It is well-known that automata minimization
 %\footnote{For complete and for partial automata, as a partial automaton could be completed by adding a sink state which
 %still appears in the mini
 could be performed
 in polynomial time~\cite{Hopcroft1971AnNL,HopUll79}.
 Recall that an automaton is minimal if and only if all states are reachable
 and pairwise distinguishable.
 The usual algorithms~\cite{Hopcroft1971AnNL,HopUll79} work in two phases: 1) only retain those states
 reachable from the start state and then 2) merge states that are not distinguishable.
 We will only use the second merging step of the minimization procedure.
 Construct the automaton $\mathcal P_2 = (\Sigma, Q_2, \delta, s_0, F)$ 
 with $F = \{ \{q\} \mid q \in Q \}$, $s_0 \in Q_2$ arbitrary and 
 $Q_2 = \{ \{q, q'\} \mid q,q'\in Q\}$, i.e., we only take the $2$-subsets and the singleton sets of $Q$ as states. This automaton has $\binom{|Q|}{2} + |Q|$
 states and could be constructed in polynomial time.
 Note that, as $|\delta(S, w)| \le |S|$ for any $S \subseteq Q$,
 the automaton is complete and we have $\delta(Q_2, x) \subseteq Q_2$ for $x \in \Sigma$.
 Then, run the second step of the minimization algorithm to find out if any
 $2$-sets\footnote{Note that all sets in $F$ for this special automaton
 are not distinguishable, and so are definitely merged by the algorithms. Alternatively, we could
 merge them beforehand, pass this input to the algorithm and see if any states at all 
 are merged.} could be merged, if this is not the case
 then they are all distinguishable. And, as $\delta(Q_2, x) \subseteq Q_2$
 for any $x \in \Sigma$,
 a pair state $\{p,q\}$ is distinguishable in $\mathcal P_2$
 if and only if it is distinguishable in $\mathcal P_{\mathscr A}$. $\qed$
 
 \medskip
 
 Alternatively, I sketch a direct argument without citing another algorithm.
 Note that, in general, for an automaton $\mathscr A = (\Sigma, Q, \delta, s_0, F)$,
 two states $\{ p, q \}$ are distinguishable~\cite{HopUll79} if and only if
 \begin{enumerate}
 \item $|\{p, q \}\cap F| = 1$, or
 \item there exists $x \in \Sigma$ such that $\{ \delta(p, x), \delta(q, x) \}$
  are distinguishable.
 \end{enumerate}
 This directly lends itself to a naive algorithm, running to the following scheme
 applied to $\mathcal P_2$ constructed above.
 First, fill a table with $n^4 + n$ (or $\binom{\binom{n}{2}}{2}$ instead of $n^4$ for the $2$-sets of $2$-sets) 
 entries, with an entry for each $2$-set of two states from $\mathcal P_2$ (note these might be $2$-sets)
 which saves if both states are distinguishable from each
 other.
 Initially, mark all $2$-sets containing precisely one final state.
 Then, run multiple passes, and in each pass, test
 for each unmarked $2$-set $\{\{p_1,p_2\},\{q_1,q_2\}\}$ and letter $x \in \Sigma$,
 if $\{ \{\delta(p_1,x),\delta(p_1,x)\}, \{\delta(q_1,x),\delta(q_2,x)\}\}$
 is marked, and if so mark $\{\{p_1,p_2\},\{q_1,q_2\}\}$.
 In each pass, we have to perform at most $O((n^4+n)|\Sigma|)$ many such tests.
 Now, if in some such pass no pair is marked, we are finished, as inductively
 any additional pass will not add any new markings.
 The non-marked state pairs are precisely the non-distinguishable ones.
 Also, in each pass, if another pass should be performed, at least
 one pair must be marked. Hence, we have to run at most $n^4+n$ many passes
 before it has to stabilize or no other pairs, i.e., all are marked, are left for marking.
 Then, when the algorithm has finished, we only have to check if all
 entries representing a $2$-set of two $2$-sets (!) are marked do conclude
 that all $2$-sets in $\mathcal P_{\mathcal A}$ are distinguishable. Obviously, this scheme
 leaves much room for improvement.

\end{proof}

Next, we state a simple observation. % beispiel, wo umkehrung nicht gilt.

\begin{lemmarep} 
\label{lem:SCC_Syn_max_implies_compl_reach}
 Let $\mathscr A = (\Sigma, Q, \delta)$
 be a strongly connected semi-automaton. 
 If $\Syn(\mathscr A)$ has maximal state complexity, then
 $\mathscr A$ is completely reachable.
\end{lemmarep}
\begin{proof}
 If the set of synchronizing words has maximal state complexity, then
 by the definition of the automaton $\mathcal P_{\mathscr A}$
 every subset $S \subseteq Q$ with $|S| > 1$ is reachable
 and at least one singleton set is reachable. But as $\mathscr A$
 is assumed to be strongly connected, then every other singleton set
 is reachable from the one supposed to be reachable. 
 So, taken together, every non-empty subset of states is reachable, i.e.,
 $\mathscr A$ is completely reachable. \qed
\end{proof}

\section{Permutation Groups and State Complexity of $\Syn(\mathscr A)$}
\label{sec:grp_case}

In Section~\ref{sec:prim_k_hom}, we state characterizations
of the $k$-homogeneous and of primitive permutation groups
by inspection of the resulting transformation monoid when non-permutations
are added. This is in the spirit of ``detecting properties of (permutation groups)  with
functions'' as presented in~\cite{ArCamSt2017}. In Section~\ref{sec:sync_max},
we introduce the sync-max permutation groups, defined by stipulating
that the resulting semi-automaton\todo{Zusammenhang Cayley-graph bzw wenn G
regulär operiert? Non-Permutation ist ja nicht teil von dem Cayley-Graph, also nicht von Monoid.}, 
given by the generators of the permutation group and a non-permutation,
is synchronizing %\footnote{damit diese teilmenge von sync gruppen; ne}
and its set of synchronizing words
has maximal state complexity.
% sync perm group die nicht sync-maximal ist
% ein zyklus, und eine non-permutation aus dem centralizer -> aber der enthält
% doch nur permutationen... wollte auf kommutativ hinaus.
% nilpotent
This definition is independent of the choice of generators
for the group and we will show that the resulting class
of groups is properly contained between the $2$-homogeneous
and the primitive groups.

\subsection{Primitive and $k$-Homogeneous Permutation Groups}
\label{sec:prim_k_hom}

%
% Todo: Noch ein Beispiel, primitive G, und G \cup f comple reachable, aber nicht 2^n - n sc.
% (0 1 ... 6) ?
% bzw p-zyklus mit ganz vielen 2-set orbits
% f so bauen, dass in einem zyklus ohne {0,1} bleibt. gibt nur zwei mengen mit 0 in so einem 2-set orbit
% also damit ein teil von f schon festgelegt.
%
% im 5-Zyklus beispiel {0,2} und {0,3} auf {1, 2^f} und {1, 3^f} also {2^f, 3^f} = {3,4}.
% belibt noch 4^f = 2, dann {4,1}^f = {1,2} geht nicht.

%
% 
% bildet vereinigung der 2-sets, welche nicht auf singleton gehen, eine partition?
% bzw deise 2-sets äquivalenzrelation?

First, as a warm-up, we state two equivalent conditions to $k$-homogeneity.
They are not hard to see, 
but seem to be unnoticed or at least never stated in the literature before.

\begin{lemmarep}
\label{lem:k_hom_char}
 Let $G \le \mathcal S_n$. The following conditions are equivalent:
 \begin{enumerate}
 \item \label{cond:hom} $G$ is $k$-homogeneous,
 \item \label{cond:any_f} for any map of rank $n - k$, in $\langle G \cup \{f\} \rangle$
  every subset of size $n - k$ is reachable,
 \item \label{cond:some_f} there exists a map of rank $n - k$ such that in $\langle G \cup \{ f \} \rangle$
  every subset of size $n - k$ is reachable.
 \end{enumerate}
\end{lemmarep}
\begin{proof}
 \begin{enumerate}
 \item First, we show that Condition~\ref{cond:hom}
  implies Condition~\ref{cond:any_f}.
  Suppose $G$ is $k$-homogeneous and let $f \colon [n] \to [n]$
  be of rank $n - k$. Hence, as we can reach at least one
  subset of size $n - k$ and as a group is $k$-homogeneous
  if and only if it is $(n-k)$-homogeneous\todo{cite huppert, oder easy to ssee, oder in preliminiaries erwähnen}, we find that we can reach every subset of size $n - k$.
  Moreover, every such subset $A\subset [n]$ could be written as $A = g(f([n]))$
  for some $g \in G$.
  
 \item That Condition~\ref{cond:any_f} implies Condition~\ref{cond:some_f}
  is obvious.
  
 \item Lastly, we show that Condition~\ref{cond:some_f}
  implies Condition~\ref{cond:hom}.
 Let $f \colon [n] \to [n]$ have rank $n - k$ such that in $\langle G \cup \{f\} \rangle$
 every subset of size $n - k$ is reachable.
 Suppose $A \subseteq [n]$ is
 a subset with $k$ elements, then
 we have some mapping $g$ in the transformation monoid generated
 by $G$ and $f$ such that $g([n]) = [n] \setminus A$.
 Write\footnote{We use the convention that for $f \colon A \to B, g \colon C \to B$,
 the composition $fg \colon C \to A$ is the function $(fg)(x) = f(g(x))$
 for $x \in C$.} 
 $g = u_1 f u_2 f \cdots u_m f$
 with $u_i \in G$ and $m \ge 1$. %, $i \in \{1,\ldots, m\}$, a (possibly empty) product of elements from $G$.
 We will show that the element $u_1$ already has to map $[n] \setminus f([n])$
 onto $A$.
 %In general, for $h \in G$ we have $|h(f((\Omega))| = n - k$ 
 %and for any subset $B \subseteq \Omega$ with $|B| \le n - k$
 %we have $|f(B)| = n - k$ if and only if $B = \{ j \} \cup\{k + 2,\ldots, n\}$
 %with $j \in \{1,\ldots, k + 1\}$.
 Now $|g([n])| = n - k$ % rank-equation zitieren?
 implies $|(fu_2f\cdots u_mf)([n])| = n - k$.
 But for any $B \subseteq [n]$, if $|f(B)| = n - k$,
 then $f(B) = f([n])$.
 Hence we must have $(fu_2f\cdots u_mf)([n]) = f([n])$, i.e.,
 we must hit the full image of $f$.
 So, 
 $$
  [n]\setminus A = g([n]) = u_1( (fu_2f\cdots u_mf)([n]) ) = u_1(f([n])).
 $$
 As $u_1$ is a permutation, $A = u_1([n] \setminus f([n]))$.
 So, for any $k$-subset $A$ of $[n]$ we find an element in $G$
 mapping $[n] \setminus f([n])$ onto $A$, i.e., $G$ is
 transitive on the $k$-sets or $k$-homogeneous.
 \end{enumerate}
 So, all conditions are equivalent. \qed
\end{proof}

\begin{remark}
 As a group is $k$-homogeneous if and only if it is $(n-k)$-homogeneous,
 we could also add that $k$-homogeneity is equivalent to the property
 that adding any (or some) function whose image contains precisely $k$
 elements gives a transformation monoid in which every $k$-subset
 is reachable.
\end{remark}

\begin{toappendix}
 % wird jetzt nichtmehr wirklich gebraucht?
For reference, we state the following.\todo{liviing stone wagner auch staten, weil später
gebraucht, in preliminaries?}

\begin{lemma}\label{lem:2-hom_implies_transitive}
 If $G$ is $2$-homogeneous with $n > 2$, then $G$ is transitive.
\end{lemma}
\begin{proof}
 Every $2$-homogeneous permutation group is primitive,
 as if we have an equivalence relation preserved by $G$, if any two distinct
 pairs are in relation to each other, we can map them to any other two pairs, i.e.,
 all elements are in relation to each other. 
 Hence, for $n > 2$, as every primitive group is transitive, $G$
 is transitive.
 This is somehow the standard argument, or arguing\todo{erklärungen auslagern.}
 directly with orbits: it is easy to see that no orbit has size one, as these
 are precisely the fixed points of the group, and so if we have two orbits
 we can take two elements from one orbit and map them onto a $2$-sets
 with one element from each orbit, which should not be possible
 as orbits form a system of blocks~\cite{cameron_1999}.
\begin{comment} 
 But with Lemma~\ref{lem:k_hom_char}, we can give another
 easy proof, for if $G$ is $2$-homogeneous and $f \colon [n] \to [n]$
 has rank $n - 1$, then we find $a,b \in [n]$
 with $f(a) = f(b)$.
 
 % a,b beliebig, da 2-hom, element g davorschalten.
 % will menge ohne c, sei c = f(d), d \notin {a,b}
 % dann auf [n] \ {f(b), d} urbild enthält nicht c.
 
 % todo, bezeichner besser wählen!
 Let us show that we can reach any subset of size $n - 1$. % der Trick auch um k-hom impliziert k - 1 hom?
 Let $r \in [n]$ and set $A = [n] \setminus \{r\}$.
 Choose $c,d \in [n] \setminus\{r\}$ with $c \ne d$
 and $g \in G$ with $g(\{c,d\}) = \{a,b\}$.
 Then $(fg)(c) = (fg)(d)$ and $fg$ has rank $n - 1$.
 Also, choose $s,t \in f([n])$ with $s \ne t$
 and $h \in G$ with $h(\{s,t\}) = \{a,b\}$.
 Then $fhf$ has rank $n - 2$.
 As $G$ is $2$-homogeneous, by Lemma~\ref{lem:k_hom_char}
 we know we find $g_3 \in \langle G \cup \{fhf\}\rangle$
 such that $(g_3fhf)([n]) = [n]\setminus \{ (fg)(d), f(r) \}$.
 Then $A = (fg)^{-1}(g_3fhf)([n])$.
 
 $|(fg)(A)| = 1$, fehlt nur $f(r)$. ne quatsch $(fg)(d)$ muss ja drinne
 sein damit urbild größer wird.
 
 %
 % Angenommen 1 geht nicht auf 2
 %
 % f([n]) = {3,...,n} von dort jede teilmenge erreichbar (2-homogen)
\end{comment}
 \qed
\end{proof}
\end{toappendix}

An automaton is synchronizing precisely if its transformation monoid contains a constant map. 
Hence, the next result is a strengthening of Theorem~\ref{thm:rystsov}.

\begin{theorem}
\todo{Completely Reachable für Transformationsmonoid gar nicht definiert!}
\label{thm:primitive_iff_compl_reach}
 A finite permutation group $G \le \mathcal S_n$
 with $n \ge 3$ is primitive if and only if, for any transformation $f \colon [n] \to [n]$
 of rank $n-1$, in the transformation semigroup $\langle G \cup \{ f \} \rangle$
 we find, for each non-empty $S \subseteq [n]$, an element $g \in \langle G \cup \{ f \} \rangle$
 such that $g([n]) = S$. %, i.e., the transformation monoid is completely reachable.
\end{theorem}
\begin{proof}
 % todo. in acc erwähnen von jemand anderem. nicht transitiv klar hinschreiben
 Suppose $G$ is a permutation group on $[n]$ that is not primitive. 
\begin{comment}
 %If $G$ is trivial, then $\langle G \cup f \rangle = \langle f \rangle$ and
 %this transformation subgroup is never completely reachable, for any function $f \colon [n] \to [n]$.
 %If $G$ is not transitive and non-trivial, then the orbits give a non-trivial partition respected
 %by $G$. If $G$ is transitive, but not primitive, we also find a non-trivial partition.
 %In both cases, let $\pi$ be a non-trivial partition of $[n]$
 Let $\pi$ be a non-trivial partition of $[n]$
 respected by $G$. Take any two distinct points $p,q$ in a class of $\pi$ and define 
 a map $f \colon [n] \to [n]$ by letting
 \[
  f(x) = \left\{ 
   \begin{array}{ll}
    q & \mbox{if } x = p, \\
    x & \mbox{if } x \ne p.
   \end{array}
  \right.
 \]
 Then $f$ is a transformation of rank $n-1$ which acts as the identical map on the quotient set $\Omega / \pi$.
 Hence the transformation monoid generated by $G$ and $f$ acts on $\Omega / \pi$
 exactly as does $G$. In particular, if $A = g([n])$ for $g \in \langle G \cup \{ f  \}\rangle$,
 then $|A| \ge |\Omega / \pi|$, so that no subsets of cardinality strictly
 less than $|\Omega / \pi|$ are reachable.
\end{comment} 
 Then, $n \ge 3$ and by Theorem~\ref{thm:rystsov}, for some $f \colon [n] \to [n]$
 of rank $n - 1$ the group $\langle G \cup \{f\}\rangle$
 does not contain a constant map. %, i.e., no map having a singleton subset as image. %So it does not reach singleton subsets of $[n]$.
 So, no element in $\langle G \cup \{f\}\rangle$ can map $[n]$ to a singleton subset of $[n]$.

 Conversely, let $G$ be a primitive permutation group on $\Omega$. Take any transformation
 monoid $M$ generated by $G$ and a transformation of rank $n - 1$.
 If $(p,q)$ is an edge of the graph $\Gamma_1(M)$, then for each $g \in G$, the 
 pair $(g(p), g(q))$ also constitutes an edge of $\Gamma_1(M)$. %btw zusammenhang 2-closure
 Since $G$ is transitive, we see that every element in $[n]$
 has an outgoing edge in $\Gamma_1(M)$. This clearly implies that $\Gamma_1(M)$
 has a directed cycle, and by the definition % genauer todo?
 of $\Gamma_1(M)$ this cycle is
 not a loop.
 Assume $\Gamma_1(M)$ is not strongly connected.
 The partition of $\Gamma_1(M)$ into strongly connected components induces a partition
 $\pi$ of $[n]$ which is nontrivial, since $\Gamma_1(M)$ has a directed cycle which
 is not a loop. As $G$ preserves the edges of $\Gamma_1(M)$, if 
 we have a path between any two vertices, we also have a path between their images.
 Hence $G$ respects $\pi$, which is not possible as $G$ is primitive by assumption.
 So $\Gamma_1(M)$ must be strongly connected. Then
 $M$ is completely reachable by Theorem~\ref{thm:Gamma_1_scc}.~$\qed$
\end{proof}

\begin{remark}%\footnotesize
 For finite permutation groups, Theorem~\ref{thm:rystsov}, Theorem~\ref{thm:Gamma_1_scc}, 
 Higman's orbital graph characterization
 of primitivity and the fact that, for finite orbital graphs, connectedness
 implies strongly connectedness (please see~\cite{cameron_1999}
 for these notions) % hier ein komma? todo, eher nicht: https://www.englisch-hilfen.de/grammar/relativsaetze.htm
 % https://dictionary.cambridge.org/grammar/british-grammar/that
 could be used to give another proof of Theorem~\ref{thm:primitive_iff_compl_reach}.
\end{remark}

\begin{remark}
\label{rem:idempotent}
 Theorem~\ref{thm:primitive_iff_compl_reach} could actually be strengthened\footnote{I am thankful to an anonymous referee for this observation.} by assuming $f : [n]\to [n]$ to be idempotent. 
 For let $G \le S_n$ be primitive and suppose the statement is valid for idempotent transformations only.
 Then, let $f : [n] \to [n]$ be any transformation. 
 Assume $f(a) = f(b)$ with $a,b \in [n]$. By Lemma~\ref{lem:k_hom_char}, $G$ is transitive.
 Hence, there exists $g \in G$ such that $a \notin g(f([n]))$
 and $gf$ permutes $[n] \setminus \{a\}$. Then, some power
 of $gf$ acts as the identity on $[n] \setminus \{a\}$, i.e., is idempotent.
\end{remark}

Because our main motivation comes from the theory of automata, let us state a
variant of Theorem~\ref{thm:primitive_iff_compl_reach} formulated in terms of automata.

\begin{corollary} % diese zuordnung eines autoamten vereinheitlichen und namen geben,zu transformationsgruppe
% \mathscr_{\mathcal T} oder so.
\label{cor:primitive_iff_compl_reach}
 Let $n \ge 0$. Suppose $G = \langle g_1, \ldots, g_k \rangle \le \mathcal S_n$.
 Then $G$ is primitive if and only if for every transformation $f\colon [n] \to [n]$
 of rank $n - 1$, the semi-automaton $\mathscr A = (\Sigma, Q, \delta)$
 with\footnote{The elements of $\Sigma$ are meant to be abstract symbols.}
 $\Sigma = \{g_1, \ldots, g_k, f \}$, $Q = [n]$ and $\delta(i, g) = g(i)$ for $i \in Q$
 and  $g \in \Sigma$ is completely reachable.
\end{corollary}

%Next, for $2$-homogeneous permutation groups, we can show a similar characterization
%to Theorem~\ref{thm:primitive_iff_compl_reach}, i.e., said intuitively, a weaker form of complete
%reachability but a stronger statement than merely Lemma~\ref{lem:k_hom_char}.

As a last consideration in this subsection, and in view of Theorem~\ref{thm:primitive_iff_compl_reach}, let us derive a sufficient condition
for ``almost complete'' reachability of all subsets of size strictly smaller
than $n - 1$ when adding any function of rank $n - 2$. By Lemma~\ref{lem:k_hom_char},
any such condition must imply $2$-homogeneity.\todo{Für einermengen gleich primitiv? Sollte gleicher Beweis gehen.
ist das in dem fall nicht die bedingung in definition der kanten vom gamma1 graphen?
also dass man für rank $n-1$ beliebige andere rank $n-1$ erzeuge (dahinterschalten bei $f$)
kann mit bestimmten zielpunkt der zweierbild hat (also denn von f mit zweierbild
auf beliebigen anderen abbilden),
also $\{a\}$ excluded, also wenn $f(a) = f(b)$ und $c$ nicht im bild,
dann kante mit f von $c$ zu $f(a)$, und g wählen dass $d$ nicht im bild und $c$ auf $e$,
dann $gf$ (erst $f$, dann $g$) liefert kannte von $c$ zu $e$, dann Gamma $1$-graph
vollständig}

\begin{propositionrep} % zeigen, 2-hom impliziert transitiv. oder livingstone/wagner in preliminaries
\label{prop:2-hom_all_subsets_n-2} 
%  A permutation group $G \le S_n$ is $2$-homogeneous
%  if and only if for any $f : [n] \to [n]$
%  of rank $n - 2$, in $\langle G \cup \{f\}\rangle$
%  every non-empty subset of size at most $n - 2$
%  is reachable. 
 Let $G \le \mathcal S_n$ and $n \ge 3$. Suppose
 the following holds true:
 % a \in g(A) in [n] \setminus {b}
 %
 % a\equiv a^G
 % 
 \begin{quote}
 \label{cond:A_include_exclude} %For any arbitrary but fixed $2$-subset $\{a,b\} \subseteq [n]$
   For any $2$-subset $\{a,b\} \subseteq [n]$
   and any $A \subseteq [n]$
   with $1 \le |A| \le n - 2$
   and $c \in [n] \setminus\{a,b\}$ we find $g \in G$ such that
   $
    \{c\} \subseteq g(A) \subseteq [n]\setminus \{a,b\}.  
   $ % erinnert an definierende bedingung in Gamma_1 graphen, damit reachability für
   % alle mengen mit non-permutation buchstaben von maximalen rank n - k
   % formulieren?
   %
   % oder auch ein "höherer" graph für charakterisierung.
  \end{quote}
  Then, for any function\todo{colon überall.} $f\colon [n] \to [n]$ of rank $n - 2$,
  in $\langle G \cup \{f\} \rangle$ every non-empty subset of size at most $n - 2$
  is reachable. In particular, by Lemma~\ref{lem:k_hom_char},
  $G$ is $2$-homogenoous.
\end{propositionrep}
\begin{proof}
 Suppose $G \le \mathcal S_n$ fulfills the condition.
 Let $\{a_1, a_2\}$ and $\{b_1, b_2\}$
 be two $2$-subsets of $[n]$.
 Set $A = [n]\setminus \{a_1, a_2\}$,
 then we find $g \in G$
 such that $g(A) = [n] \setminus\{b_1, b_2\}$.
 Hence $g(\{a_1, a_2\}) = \{b_1, b_2\}$
 and $G$ is $2$-homogeneous.

 Now, let $f \colon [n] \to [n]$ be a function of rank $n - 2$
 and suppose for $a,b,c,d \in [n]$ we have 
 $f([n]) = [n] \setminus \{a,b\}$ and $f(c) = f(d)$, $c \ne d$.
 By Lemma~\ref{lem:k_hom_char}, as $G$ is $2$-homogeneous,
 the subsets of size $n - 2$
 are reachable. Next, we argue inductively.
 Let $A \subseteq [n]$ be any non-empty subset
 with $|A| < n - 2$ and assume all subsets of size $|A| + 1$ to $n - 2$
 are reachable.
 By assumption, we find $g \in G$
 such that
 $$
  \{f(c)\} \subseteq g(A) \subseteq [n] \setminus \{a,b\}.
 $$
 We have $|f^{-1}(g(A))| \in \{|A|+1, |A|+2\}$.
 Set
 $$
  B = \left\{ 
  \begin{array}{ll}
   f^{-1}(g(A)) \setminus \{c\} & \mbox{if } |f^{-1}(g(A))| = n - 1; \\ 
   f^{-1}(g(A))                 & \mbox{otherwise.}
  \end{array}
  \right.
 $$
 As $\{c,d\} \subseteq f^{-1}(g(A))$ and $f(c) = f(d)$
 we have $f(B) = f(f^{-1}(g(A)))$.
 Inductively, we find $h \in \langle G \cup \{f\} \rangle$
 such that $h([n]) = B$.
 As $g(A) \subseteq f([n])$
 we have
 $$
  g(A) = g(A) \cap f([n]) = f(f^{-1}(g(A))) = f(B) = f(h([n])).
 $$
 and so $A = g^{-1}(f(h([n])))$. \qed
\end{proof}

\subsection{Sync-Maximal Permutation Groups}
\label{sec:sync_max}

%Let $M$ be some transformation monoid. Note that if $M$ is transitive (in particular, if it contains
%a transitive permutation group), then 

Here, we introduce the sync-maximal permutation groups. For their definition,
we associate to a given group and a non-permutation %function of rank $n - 1$
a semi-automaton whose letters are generators
of the group and the non-permutation. This definition is actually independent
of the choice of generators of the group.
But before giving the definition of sync-maximal groups, let us
first state a result linking the notion of complete reachability
to the state complexity of the set of synchronizing words.
%Our first result combines Proposition~\ref{prop:n-1_reachable_iff_transitive_perm_grp}
%and Lemma~\ref{lem:SCC_Syn_max_implies_compl_reach}.

\begin{comment} 
The automata constructed in Corollary~\ref{cor:primitive_iff_compl_reach}
are always strongly connected, as every primitive permutation group is transitive.
Hence, the condition that $\Syn(\mathscr A)$ has maximal state complexity
for these automata is a stronger restriction than only asserting complete reachability.
Every $2$-homogenous permutation group of degree strictly greater than two is primitive, %dixon/mortimer zitieren?
so this is a stronger restriction on the group in question.
\end{comment}

\begin{proposition}
\label{prop:groups_max_sc_implies_compl_reach} 
  Let $G = \langle g_1, \ldots, g_k \rangle \le \mathcal S_n$ be a permutation group and $f \colon [n] \to [n]$
  be a non-permutation. Set $\Sigma = \{g_1, \ldots, g_k, f\}$ and $\mathscr A = (\Sigma, [n], \delta)$
  with $\delta(m, g) = g(m)$ for $m \in [n]$ and $g \in \Sigma$.
  If $n > 2$ and $\stc(\Syn(\mathscr A)) = 2^n - n$, then $G$ is transitive
  and $\mathscr A$ completely reachable.
\end{proposition}
\begin{proof}
  %As we have some non-permutation, we must have $n \ge 2$.
  Suppose $n > 2$. As $\stc(\Syn(\mathscr A)) = 2^n - n$, in $\mathcal P_{\mathscr A}$
  sets of size $n - 1$ are reachable\footnote{This argument only works for $n > 2$.
  If $n = 2$, as singletons sets are not distinguishable, if the state complexity
  is maximal, not all singleton sets need to be reachable from $Q$. Also, see
  Remark~\ref{remark:sc_implies_compl_reach}.}.
  Hence, $f$ must have rank $n - 1$.
  Then, by Lemma~\ref{lem:k_hom_char}, the group $G$ is transitive, which 
  implies the semi-automaton $\mathscr A$ is strongly connected.
  So, by Lemma~\ref{lem:SCC_Syn_max_implies_compl_reach},
  the semi-automaton is completely reachable. \qed
\end{proof} 

\begin{remark}%\footnotesize
\label{remark:sc_implies_compl_reach}
 For $n = 2$, if $G$ only contains the identity transformation, then adding
 any non-permutation gives an automaton such that the set of synchronizing
 words has state complexity two, but it is not completely reachable nor is $G$
 transitive.
 Also, note that the assumption $\stc(\Syn(\mathscr A)) = 2^n - n$
 implies that $f$ must have rank $n - 1$, for otherwise sets of size $n - 1$
 are not reachable.
\end{remark}

\begin{definition}
\label{def:sync-maximal}
  A permutation group $G = \langle g_1, \ldots, g_k \rangle \le \mathcal S_n$ 
  is called \emph{sync-maximal}, if for any map $f \colon [n] \to [n]$ of rank $n - 1$,
  for the automaton $\mathscr A = (\Sigma, [n], \delta)$
  with $\Sigma = \{g_1, \ldots, g_k, f\}$ and
  $\delta(m, g) = g(m)$ for $m \in [n]$ and $g \in \Sigma$,
  we have $\stc(\Syn(\mathscr A)) = 2^n - n$,
\end{definition}

% Because the next result only applies for $n > 2$, and to make the sync-maximal permutation 
% a subclass of the primitive permutation group, we have chosen to assert $n > 2$
% in the definition.

As written, the definition involves a specific set of generators for $G$.
But the resulting transformation monoids are equal for different generators
and we can write one set of generators in terms of another.
So reachability of subsets and distinguishability of subsets is preserved
by a change of generators. Hence, the definition is actually independent of the specific choice of generators for $G$.
This might be different if we are concerned with the length of shortest words to reach
certain subsets, but that is not part of the definition.

\begin{proposition}
\label{cor:max_sc_implies_primitive}
  %Let $G = \langle g_1, \ldots, g_k \rangle \le S_n$ be a permutation group.
  %If for any map $f : [n] \to [n]$ of rank $n - 1$,
  %for the automaton $\mathscr A = (\Sigma, [n], \delta)$
  %with $\Sigma = \{g_1, \ldots, g_k, f\}$ and
  %$\delta(m, g) = g(m)$ for $m \in [n]$ and $g \in \Sigma$,
  %we have $\stc(\Syn(\mathscr A)) = 2^n - n$, then $G$ is primitive.
  Every sync-maximal permutation group is primitive.
\end{proposition}
\begin{proof} 
 By our definition, every permutation group of degree $n \le 2$
 is primitive. Note that this might be different in the literature, but 
 is in concordance with~\cite{DBLP:journals/tcs/ArnoldS06}.
 If $n > 2$, then, by Proposition~\ref{prop:groups_max_sc_implies_compl_reach},
 the automaton, associated to $G$ and a function of rank $n - 1$ as in Proposition~\ref{prop:groups_max_sc_implies_compl_reach}, is completely reachable. Hence, by Theorem~\ref{thm:primitive_iff_compl_reach},
 the group $G$ is primitive\footnote{Alternatively, by noting that $\stc(\Syn(\mathscr A)) = 2^n - n$
 implies that the automaton $\mathscr A$ has at least one  synchronizing word, i.e.,
 the transformation monoid admits a constant map. Then, invoking
 Theorem~\ref{thm:rystsov} gives primitivity.}.~\qed
\end{proof}

\begin{remark}%\footnotesize
 By case analysis, note that for $n \le 2$ every group is sync-maximal.
 But also for $n \le 2$, by our definition of primitivity,
 every permutation group is primitive. See the explanations in Section~\ref{sec::preliminaries}.
\end{remark}

Next, we relate the condition that the set of synchronizing words has maximal state complexity
to the notion of $2$-homogeneity. However, as shown in Example~\ref{ex:2-hom},
we do not get a characterization of $2$-homogeneity similar to Theorem~\ref{thm:primitive_iff_compl_reach}
for primitivity.

\begin{comment} 
intuition, immer ein f genau ein element weniger.
\begin{lemma}
 Let $G = \langle g_1, \ldots, g_k \rangle \le \operatorname{Sym}([n])$
 with $n \ge 1$. Suppose $S \subseteq [n]$, $f \colon [n] \to [n]$ has rank $n-1$,
 $x \in \langle G \cup f \rangle$
 and consider $T = S^x$. Set $\Sigma = \{ g_1, \ldots, g_k, f \}$.
 Then, for every representing word $w \in \Sigma^*$ % mixture der notation aus formal langauge und group theory [hoch^x, aber keine punkte für elemtne, nich talphab, beta \Omega ]
 of $x$ we have $|x|_f \ge |S| - |T|$. % mehr geht, wenn es ichts macht

 % für 2-homogenous findet man eins mit genau f vielen jeweils!, das
 % ist wichtig dass man dann unterschiedliche kardinalitäten unterscheiden kann.
\end{lemma}
\begin{proof} Todo 
\end{proof}

\begin{lemma}
 Let $G = \langle g_1, \ldots, g_k \rangle \le \operatorname{Sym}([n])$
 be $2$-homogeneous
 with $n \ge 1$. Suppose $S \subseteq [n]$, $f \colon [n] \to [n]$ has rank $n-1$,
 $x \in \langle G \cup f \rangle$
 and consider $T = S^x$. Set $\Sigma = \{ g_1, \ldots, g_k, f \}$.
 Then, we find a representing word $w \in \Sigma^*$
 of $x$ with $|x|_f = |S| - |T|$.
\end{lemma}
\begin{proof} 
\end{proof}
\end{comment}

\begin{proposition}
\label{prop:2_hom_perm_group}
%  Suppose $G = \langle g_1, \ldots, g_k \rangle \le \mathcal S_n$. %\operatorname{Sym}([n])$
%  is a $2$-homogenous permutation group. Then for every transformation $f : [n] \to [n]$
%  of rank $n - 1$, for the semi-automaton $\mathscr A = (\Sigma, Q, \delta)$
%  with $\Sigma = \{g_1, \ldots, g_k, f \}$, $Q = [n]$ and $\delta(m, g) = g(m)$
%  for $g \in \Sigma$, we have $\stc(\Syn(\mathscr A)) = 2^n - n$.
% A $2$-homogeneous permutation group of finite degree is sync-maximal.
 If $G \le \mathcal S_n$ is $2$-homogeneous, then $G$ is sync-maximal.
\end{proposition} % definieren, singlton final, und alle äquialentt, power automat, ext paper + malinkov
\begin{proof} 
 Assume $n > 1$, otherwise the statement is trivially true.
 Suppose $G = \langle g_1, \ldots, g_k \rangle$
 is $2$-homogeneous and let $f \colon [n] \to [n]$ be of rank $n - 1$.
 Without loss of generality, assume $f(0) = f(1)$.
 As every $2$-homogeneous permutation group is primitive~\cite{DBLP:journals/tcs/ArnoldS06},
 by Theorem~\ref{thm:primitive_iff_compl_reach},
 the semi-automaton $\mathscr A$ is completely reachable.
 By Lemma~\ref{lem:compl_reach_implies_max_Syn_2sets},
 it is enough to show that all $2$-sets are distinguishable.
 Let $\{q_1, q_2\}, \{q_3, q_4\} \subseteq Q$ be two distinct $2$-sets.
 By $2$-homogeneity of $G$, we find $g \in G$
 such that $g(\{q_1, q_2 \}) = \{0,1\}$. Then $g(\{ q_3, q_4 \}) \ne \{0,1\}$. %todo functionksnotation vereinheitlichen
 Hence, $|f(g(\{q_1, q_2 \}))| = 1$ 
 and $|f(g(\{q_3, q_4 \}))| = 2$, and the function $fg$ could be written as a word over 
 the generators of $G$ and $f$. $\qed$
\end{proof}

%
% Gegenbeispiel, cerny-automat? für n odd
% zumindest mit dem beweisschema f(0) = f(1)=1 sonst identität
% also noch andere f's angucken.
% 
% primzahlzyklus, f dazu, dann 2^n - n sc?
% 
% (0 1 2) fälle durchrehnen. obdA f(0) = f(1)
%  
% img(f) 3 möglichkeiten, dann ob es tauscht oder nicht, also 6 möglichkeiten.
%
% img(f) = {1,2} und identität
%  {1,2}->{2,0}->{0,1} ist 2-homge
%
% (0 1 2 3 4)
%
% {1,2}->{2,3}->{3,4}->{4,0}->{0,1}->{1,2}
% aber 10 2-teilmengen, z.B. {1,3} nicht drin.
% 
% {4,0}f = {4,1}, dann weiter von da. beispiel, dass alle möglichkeiten von f, d.h. in prop weiter unten
% für binär keine einschränkungen.
%
% alle 2-er mengen unterscheidbar?
% obige bahn hat {0,1} drin, die in anderer bahn davon unterscheidbar. Untereinander? 
% bei obiger bahn klar, da verschiedene potenz bis zu {0,1}
%
% {1,3}->{2,4}->{3,0}->{4,1}->{0,2}->{1,3}
%
% mit {3,0}f = {3,1} in gleicher bahn
% {0,2}f = {1,2}, einzige, damit auch klar.
%
% erfüllt jedes f dass man zwischen 2-set bahnen abbilden kann?
%
% f als 4-string auf bild, z.B. 2340 für kleinstes element auf 2 usw
%
% {3,0}f = {3f, 1f}
% {2,0}f = {2f, 1f}, {2,1} in anderem, fälle unterscheiden, wo liegt bild von {2,1}
%
% liegt es in erster 2-set bahn, dann untere alle unterschiedbar, wenn nicht, dann obere von unteren unterscheidbar.
%
% {1,2} und {1,3}

The next example shows that the converse of Proposition~\ref{prop:2_hom_perm_group}
does not hold.

\begin{example} % keine umkehrung
%\footnotesize
\label{ex:2-hom}
Let\todo{Zyklennotation nicht eingführt.}
%$g = (0~1~2~3~4)$ 
$g \colon [5] \to [5]$ be the cycle given by $g(i) = i + 1$ for $i \in \{0,\ldots,3\}$ and $g(4) = 0$.
Set $G = \langle \{ g \} \rangle$.
Then, as it is a cycle of prime length, $G$ is primitive.
So, by Theorem~\ref{thm:primitive_iff_compl_reach}, for any $f \colon [n] \to [n]$ of rank $n - 1$
the transformation semigroup $\langle G \cup \{f\} \rangle$ is completely reachable.
Also, we show $\stc(\Syn(\mathscr A)) = 2^n - n$.
But $G$ is not $2$-homogeneous.
We have the following two orbits on the $2$-sets: 
$
 A = \{ \{ 1,2 \}, \{2,3\}, \{3,4\}, \{4,0\}, \{0,1\} \}$ and 
 $
 B = \{ \{0,2\}, \{1,3\}, \{2,4\}, \{3,0\}, \{4,1\}\}.
$
Let $f \colon [n] \to [n]$ be any map of rank $n-1$. Without loss of generality, 
we can assume $f(0) = f(1)$. Then, two distinct $2$-sets are distinguishable,
if one could be mapped to $\{0,1\}$, but not the other, as a final application of $f$
gives that one is mapped to a singleton, but not the other.
First, note that all $2$-sets in $A$ are distinguishable, as for each
$\{x,y\} \in A$ we find a unique $0 \le k < |A|$ such
that $g^k(\{z,v\})= \{0,1\}$ if and only if $\{z,v\} = \{x,y\}$
for each $\{z,v\} \in A$, as $g$ permutes $A$.
If we have any $\{x,y\} \in B$ such that
$f(\{x,y\}) \in A$, then all sets in $B$ are distinguishable.
For if $\{z,u\} \in B$ there exists a unique $g^k$
with $0 \le k < |B|$, $g^k(\{z,u\}) = \{x,y\}$
and\todo{Hier tauchte das $g^k$ noch rechts auf, steht jetzt so auf arxiv.}
$g^k(\{z', u'\}) \ne \{x,y\}$ for each other $\{z', u'\} \in B \setminus \{\{z,u\}\}$.
As $\{0,1\} \notin B$ and $f$ is injective on $\{1,2,3,4\}$
and on $\{0,2,3,4\}$, $f$ is injective on $B$.
Hence $f(g^k(\{ z', u' \}) \ne f(g^k(\{z,u\}))$ for the previously chosen $\{z,u\}\in B$
and $\{z',u'\} \in B \setminus \{ \{ z,u \} \}$.
Now, choose $g^l$ such that $g^l(f(g^k(\{z,u\}))) = \{0,1\}$.
In any case, i.e., whether $f(g^k(\{ z', u' \}))$ is in $B$ or in $A$,
we have $g^l(f(g^k(\{z',u'\}))) \ne \{0,1\}$.
%Hence, for any two distinct $2$-sets, one could be mapped to $\{0,1\}$
%and the other to any other set. A last application of $f$
% , so that we have two different
% $2$-sets in $A$, which are distinguishable. (brauchte nichtmal bis auf $\{0,1\}$
% abbilden).
So, all $2$-sets in $A$ are distinguishable and all $2$-sets in $B$.
That a $2$-set from $A$ is distinguishable from any $2$-set in $B$
is clear, as we can map the $2$-set from $A$ to $\{0,1\}$ by a power of $g$,
and the one from $B$ would not be mapped to $\{0,1\}$.
Lastly, we show that we must have some $\{x,y\} \in B$
with $f(\{x,y\}) \in A$, which gives the claim.
Consider the sets $\{0,2\}, \{0,3\}, \{1,4\}$ and $\{1,3\}$ from $B$
and suppose their images are all contained in $B$, i.e.,
$\{ \{ f(0), f(2) \}, \{ f(0), f(3) \}, \{ f(0), f(4) \} \} \subseteq B$.
But in $B$ at most two sets share an element, hence
$|\{ f(2), f(3), f(4) \}| \le 2$, which is not possible as $f$
has rank $n-1$ and we already have $f(0) = f(1)$.
So, some image of these three sets must be in $A$. $\qed$
\end{example}

\section{$k$-Reachable Permutation Groups}
\label{sec:k_reachable_groups}

% csr text nur reinkopiert, noch drüber schauen
% und motivation
% steinberg sync gruppen + übersichtsartikel zitieren.
%
% todo elemente nicht points nennen, sieeh beweis thm:primitive_compl_reach
%
% todo, transformation semigroup compl reachable definieren

%
% jede 1-reachable (=primitiv) auch k-reachable
% da nach definition auch map of rank n - k drin?
% die höherwertigen Gamma_k(M)-Graphen?
%
% G 2-reachable, f rank n - 1 dazu
%  dann ein g vom rank n - 2 drin, also alle mengen n - 2i erreichbar
% ist 2-homogenous, damit auch die von größe n - i erreichbar?

A permutation group $G \le \mathcal S_n$ is called \emph{synchronizing}, if
for any non-permutation the transformation monoid $\langle G \cup \{ f \} \rangle$
contains a constant map. This notion was introduced in~\cite{DBLP:journals/tcs/ArnoldS06}.
For further information on synchronizing groups and its relation
to the \v{C}ern\'y conjecture, see the survey~\cite{ArCamSt2017}.
In~\cite{ArCamSt2017}, the question was asked to detect properties of permutation groups
by functions. % in the way synchronizing group are defined.
Theorem~\ref{thm:primitive_iff_compl_reach} and Theorem~\ref{thm:rystsov} are in this vain.
Note that every synchronizing group is primitive~\cite{DBLP:journals/tcs/ArnoldS06},
but not conversely~\cite{Neumann09}.
By Theorem~\ref{thm:primitive_iff_compl_reach}, primitive groups have the
property that if we add any function of rank $n-1$ every non-empty subset is reachable.
%So, a natural question to ask is, if we add any function of rank $n - k$, $k > 0$,
%are all conceivable subsets, i.e., those of size $n - k, n - 2k, \ldots$
%and so on, reachable? And what groups do we get, if we assume this property? %require this property? english?
% Next, we generalize  Theorem~\ref{thm:primitive_iff_compl_reach}
% Here we look at another variant to transfer the notion of completely
% reachable to the permutation group setting in analogy with the notion
% of synchronizing permutation groups.
Motivated by this, we introduce $k$-reachable groups, which generalize
the condition of complete reachablity mentioned in Theorem~\ref{thm:primitive_iff_compl_reach}.

\begin{definition} % degree einführen
 A permutation group over the finite set $[n]$
 with $n > 1$
 is called \emph{$k$-reachable}, if for any
 map $f : [n] \to [n]$ of rank $n - k$
 all subsets of cardinality
 $$
  n - k, n - 2k, \ldots, n - \left( \lceil n / k \rceil - 1 \right)\cdot k
 $$
 are reachable, i.e., we have some transformation
 in the transformation monoid generated by $G$ and $f$
 which maps $[n]$ to any such set.
\end{definition}

By Theorem~\ref{thm:primitive_iff_compl_reach}, the $1$-reachable group are precisely the
primitive groups.
Also note that $(n-1)$-reachable is the same as transitivity.

\begin{proposition}
\label{prop:k_reachable_implies_k_hom}
 A $k$-reachable permutation group is $k$-homogeneous.
\end{proposition}
\begin{proof} 
 This is implied by Lemma~\ref{lem:k_hom_char}
 and the definition of $k$-reachability.~\qed
\end{proof}

%By our previous result, $1$-reachable groups are primitive, which is more than being transitive.
The reverse implication does not hold in the previous proposition.
For example, we find $1$-homogeneous, i.e., transitive groups, which
are not $1$-reachable, i.e., primitive by Theorem~\ref{thm:primitive_iff_compl_reach}.
By a result of Livingstone and Wagner~\cite{Huppert82,Livingstone1965},
for $5 \le k \le n/2$ a $k$-homogeneous permutation group of degree $n$
is $k$-transitive, and for $k \le n/2$
a $k$-homogeneous permutation group is also $(k-1)$-homogeneous. 
As $2$-homogeneity implies synchronizability~\cite{ArCamSt2017}, combined with the fact that $k$-homogeneity
is equivalent with $(n - k)$-homogeneity, 
a $k$-reachable group for any $1 < k < n-1$ is synchronizable
and we get the next statement together with our previous results.

\begin{proposition}
 A $k$-reachable permutation group of degree $n$ for $1 < k < n - 1$ is synchronizable. 
 For $k = 1$ we have precisely the primitive permutation groups,
 and for $k = n - 1$ precisely the permutation groups
 which are transitive in their action.
\end{proposition}

For $k \in \{2,3,4\}$ the non-$k$-transitive but $k$-homogeneous groups
where determined by Kantor \cite{Kantor1972}. A list of all possible $k$-transitive groups of finite degree for $k \ge 2$
could be found in \cite{cameron_1999}. For $k \ge 6$ the only cases
are the symmetric group or the alternating group. 
% Being $k$-reachable for every $1 \le k < |\Omega|$ might be another possible candidate for an analogous
% property of completely reachable automata (or transformation monoids) in the case of permutation groups
% in analogy to the synchronization property already in existence.
% But the previous reasoning shows that this will not be very revealing in general.
So, if we want to formulate a stronger version of Theorem~\ref{thm:primitive_iff_compl_reach}
by assuming $k$-reachability for all $1 \le k \le n$, for
$n \ge 6$ only the symmetric and the alternating group fulfill this condition.

\begin{proposition}
  If a permutation group of degree $n \ge 6$ is $k$-reachable for all
  $1 \le k \le n - 1$, then it is either the symmetric group or the alternating group. % diese begriffe noch definieren
\end{proposition}

Or to be more specific.

\begin{proposition}
 If a permutation group of degree $n$
 is $k$-reachable for $6 \le k \le n - 6$, then 
 it is either the symmetric or the alternating group.
\end{proposition}

Note that the classification of the $k$-transitive groups of finite degree for $k > 1$ cited above relies
on the classification of the finite simple groups, see~\cite{Solomon01}
for an account of this important and highly non-trivial result.

% 1-hom is transitive, aber wir haben mehr, genau primitiv, hier keien charakterisierung, beispiel oben priitive gruppen
% aber auch eins für k beliebig.
% umkehrung nicht, beispiel

% allgemein + paper kantor und so

% inklusionsbeziehungen? oder unvergleichbar?

\section{Conclusion}

 % umkehrungen, und sc() maximal in perm groups characterisieren.
 % regebnisse über binäres alphabet natürlich auch für |\Giam| \ge 2, noch hinschreiben.
 % at least binary, prop {a,b} in \Sigma 

 We have given a new characterization for primitive permutation groups.
 In an analogous way, with the property that the set of synchronizing words
 has maximal state complexity, we have introduced the class of sync-maximal 
 permutation groups.
 We have shown that the sync-maximal permutation groups are primitive and
 that $2$-homogeneous groups are sync-maximal. %, but not conversely.
 Example~\ref{ex:2-hom} shows that not every sync-maximal permutation group is $2$-homogeneous.
 However, we do not know if the converse of Proposition~\ref{cor:max_sc_implies_primitive}
 is true, i.e., does there exist a primitive permutation group that is not sync-maximal?
 % could be used to characterize a new class of
 %permutation groups. Let us call them \emph{sync-maximal} permutation groups.
 %We have shown that the \emph{sync-maximal} permutation groups are primitive and
 %that $2$-homogenous groups are \emph{sync-maximal}, but not conversely.
 More results on the structure of the sync-maximal permutation groups would be
 highly interesting and might be the goal of future investigations.
 Also, for future investigations, their relation to the synchronizing groups, as introduced in~\cite{DBLP:journals/tcs/ArnoldS06},
 is of interest, as synchronizing groups lie strictly between the $2$-homogenous
 and primitive groups~\cite{ArCamSt2017,DBLP:journals/tcs/ArnoldS06,Neumann09}.
 More specifically, let us mention that in~\cite{ArCamSt2017} a whole hierarchy
 of permutation groups was surveyed.
 The hierarchy is the following,
 stated without formally introducing
 every term: \vspace{-0.3cm}
 \begin{multline*}
     2\mbox{-transitive} \subsetneq 2\mbox{-homogeneous} 
   \subsetneq \mathbb QI 
   \subseteq  \mbox{spreading} 
   \subsetneq  \mbox{separating} \\
   \subsetneq  \mbox{synchronizing}
   \subsetneq  \mbox{primitive} 
   \subsetneq  \mbox{transitive}.
 \end{multline*}
 Note that it is unknown if the inclusion between
 the $\mathbb Q I$-groups
 and the spreading groups is proper. The question of fitting
 the sync-maximal groups more precisely
 into this hierarchy arises naturally.
%  Also, we took a closer look at completely reachable automata
%  over binary alphabets. We have given sufficient conditions
%  for completely reachable automata %containing one letter that acts as a cyclic permutation
%  %with a single orbit and another letter that collapses 
%  %precisely two states to also
%  over binary alphabets to also
%  have the property that their sets of synchronizing words have maximal
%  state complexity. Our conditions yield that this property is true for a family of automata for which
%  this was previously only conjectured.
 Lastly, we introduced $k$-reachable permutation groups.
 But for %$k \notin \{2,3,4\}$ % k = 1, k = n - 1 noch
 most $k$ these do not give any groups beside the symmetric and alternating
 groups.
 A more close investigation and characterization of these groups for $k \in \{2,3,4\}$
 is still open. %reserved for future work. % oder könnte noch untersucht werden als formulierung.
 
 \medskip
%\section*{Acknowledgement}
{\footnotesize 
\noindent\textbf{Acknowledgement:} I thank my supervisor, Prof. Dr. Henning Fernau, for giving valuable feedback, discussions and research suggestions concerning the content of this article. I also thank Prof. Dr.
Mikhail V. Volkov for introducing
our working group to the idea of completely reachable automata at a joint workshop in Trier in the spring of 2019, from which the present work draws inspiration. Lastly, the argument in the proof of Theorem~\ref{thm:primitive_iff_compl_reach}, which closely resembles a proof from~\cite{DBLP:journals/tcs/ArnoldS06},
was communicated to me by an anonymous referee of a considerable premature 
version of this work. I thereby sincerely thank the referee for this and other remarks.
I also thank another anonymous referee for Remark~\ref{rem:idempotent}
and other suggestions related to the content of this work. In case a referee is wondering
why an example from the submitted version is missing, it contained a subtle error and
I was unable to fix it in due time
for the final version.
}

% Todo, crossref geht nicht
% https://www.google.com/search?q=inproceedings+crossref&oq=inproceedings+crossref+&aqs=chrome..69i57.5663j0j7&sourceid=chrome&ie=UTF-8
% https://tex.stackexchange.com/questions/256227/when-multiple-entries-crossref-the-same-proceeding-the-proceeding-details-are-n
% https://tex.stackexchange.com/questions/148887/use-crossref-in-bibliography-but-only-show-the-full-main-entry/148978
% https://www.google.com/search?ei=SooDX77JGpmI1fAPo92G4A0&q=crossref+not+working+latex&oq=crossref+not+working+latex&gs_lcp=CgZwc3ktYWIQAzIFCAAQzQI6BAgAEEc6BggAEBYQHjoFCCEQoAFQwDBY7zpg-DtoAHABeACAAWyIAbIEkgEDNS4xmAEAoAEBqgEHZ3dzLXdpeg&sclient=psy-ab&ved=0ahUKEwi-lpmnvLnqAhUZRBUIHaOuAdwQ4dUDCAw&uact=5

{
\bibliographystyle{splncs04}
\bibliography{ms} 
%\printbibliography
}

%\clearpage
%\section{Appendix}
%\input{appendix}

\end{document}